\documentclass{llncs}
\usepackage[utf8]{inputenc}
\usepackage[T1]{fontenc}
\usepackage[english]{babel}

\usepackage{mathtools}
\usepackage{amsmath,amssymb,amsfonts}
\usepackage[colorlinks,linkcolor={blue}]{hyperref}
\usepackage{cleveref}
\usepackage{enumitem}

\crefname{lemma}{Lemma}{Lemmas}
\Crefname{lemma}{Lemma}{Lemmas}
\crefname{theorem}{Theorem}{Theorems}
\Crefname{theorem}{Theorem}{Theorems}
\crefname{observation}{Observation}{Observations}
\Crefname{observation}{Observation}{Observations}
\crefname{definition}{Definition}{Definitions}
\Crefname{definition}{Definition}{Definitions}
\crefname{section}{Section}{Sections}
\Crefname{section}{Section}{Sections}
\crefname{figure}{Figure}{Figures}
\Crefname{figure}{Figure}{Figures}
\crefname{equation}{}{}
\Crefname{equation}{}{}

\renewcommand{\subset}{\subseteq}

\newcommand{\ceil}[1]{\left\lceil{#1}\right\rceil}
\newcommand{\floor}[1]{\left\lfloor{#1}\right\rfloor}

\newcommand{\abs}[1]{\left | #1 \right |}
\newcommand{\set}[1]{\left \{ #1 \right \}}

\newcommand{\Oh}{\mathcal{O}}

\newtheorem{fact}[theorem]{Fact}

\crefname{fact}{Fact}{Facts}
\DeclarePairedDelimiterX\Set[2]{\lbrace}{\rbrace}%
 { #1 \,\delimsize|\, #2 }

\def\FullBox{\hbox{\vrule width 8pt height 8pt depth 0pt}}

\def\qed{\ifmmode\qquad\FullBox\else{\unskip\nobreak\hfil
\penalty50\hskip1em\null\nobreak\hfil\FullBox
\parfillskip=0pt\finalhyphendemerits=0\endgraf}\fi}

\def\qedsketch{\ifmmode\Box\else{\unskip\nobreak\hfil
\penalty50\hskip1em\null\nobreak\hfil$\Box$
\parfillskip=0pt\finalhyphendemerits=0\endgraf}\fi}

\newcommand{\R}{{\mathbb R}}
\newcommand{\N}{{\mathbb{N}}}

\newcommand{\F}{{\mathbb F}}

\newcommand{\loglog}{{\mathop{\mathrm{loglog}}}}

\newcommand{\class}[1]{\mathbf{#1}}

\renewcommand{\P}{\class{P}}

\newcommand{\E}{\mathbb{E}}

\renewcommand{\Pr}{\mathbf{Pr}}

\newcommand\drop[1]{}

\usepackage[pdftex]{graphicx}
\usepackage[usenames,dvipsnames,table]{xcolor}
\definecolor{shade}{RGB}{235,235,235}

\usepackage{todonotes}

\usepackage{authblk}

\title{Quicksort, Largest Bucket, and Min-Wise Hashing with Limited Independence}
\author{Mathias Bæk Tejs Knudsen \inst{1} \thanks{Research partly supported by Mikkel Thorup's
    Advanced Grant from the Danish Council for Independent Research
    under the Sapere Aude programme and the FNU project AlgoDisc - Discrete
    Mathematics, Algorithms, and Data Structures.}
    \and Morten Stöckel\inst{2} \thanks{This author is supported by the Danish National Research Foundation under the Sapere Aude program.}}

\institute{University of Copenhagen \\ \email{knudsen@di.ku.dk}
\and IT University of Copenhagen \\ \email{mstc@itu.dk}}

\date{}

\begin{document}
\setcounter{page}{0}
\maketitle

\begin{abstract}
Randomized algorithms and data structures are often analyzed under the assumption of access to a perfect source of randomness. The most fundamental metric used to measure how ``random'' a hash function or a random number generator is, is its \emph{independence}: a sequence of random variables is said to be $k$-independent if every variable is uniform and every size $k$ subset is independent.

In this paper we consider three classic algorithms under limited independence. Besides the theoretical interest in removing the unrealistic assumption of full independence, the work is motivated by lower independence being more practical.
We provide new bounds for randomized quicksort, min-wise hashing and largest bucket size under limited independence. Our results can be summarized as follows.
\begin{itemize}
	\item \emph{Randomized quicksort.} When pivot elements are computed using a $5$-independent hash function, Karloff and Raghavan, J.ACM'93 showed $\Oh ( n \log n)$ expected worst-case running time for a special version of quicksort. We improve upon this, showing that the same running time is achieved with only $4$-independence.

	\item \emph{Min-wise hashing.} For a set $A$, consider the probability of a particular element being mapped to the smallest hash value. It is known that $5$-independence implies the optimal probability $\Oh (1 /n)$.
	Broder et al., STOC'98 showed that $2$-independence implies it is $\Oh(1 / \sqrt{|A|})$. We show a matching lower bound as well as new tight bounds for $3$- and $4$-independent hash functions.

	\item \emph{Largest bucket.} We consider the case where $n$ balls are distributed to $n$ buckets using a $k$
	      -independent hash function and analyze the largest bucket size. Alon et. al, STOC'97 showed that there
	      exists a $2$-independent hash function
	      implying a bucket of size $\Omega ( n^{1/2})$. We generalize the bound,
	      providing a $k$-independent family of functions that imply size $\Omega ( n^{1/k})$.

\end{itemize}
\end{abstract}
\thispagestyle{empty}
\newpage
\setcounter{page}{1}
\section{Introduction}
A unifying metric of strength of hash functions and pseudorandom number generators is the \emph{independence}
of the function. We say that a sequence of random variables is $k$-independent if every random variable is
uniform and every size $k$ subset is independent. A question of theoretical interest is,
regarding
each algorithmic
application, \emph{how much independence is required?}. With the standard implementation of a random generator
or hash function via a $k-$ degree polynomial $k$ determines both the space used and the amount of randomness
provided.
A typical assumption when performing algorithmic analysis is to just assume full independence,
i.e., that
for input size $n$ then the hash function is $n$-independent. Besides the interest from a theoretic perspective,
the question of how much independence is required is in fact interesting from a practical perspective: hash
functions and generators with lower independence are as a rule of thumb faster in practice than those with
higher independence, hence if it is proven that the algorithmic application needs only $k$-independence to
work, then it can provide a speedup for an implementation to specifically pick a fast construction that
provides the required $k$-independence.
In this paper we
consider three fundamental applications of random hashing,
where we provide new bounds for limited independence.

{\bf Min-wise hashing.} We consider the commonly used scheme \emph{min-wise hashing}, which was
first introduced by Broder~\cite{Broder97onthe} and has several well-founded applications (see \Cref{kind:related}).
Here we study families of hash functions, where a function $h$ is picked uniformly at random from the family and
applied to all elements of a set $A$. For any element $x \in A$ we say that $h$ is min-wise independent if
$\Pr( \min h(A) = x ) = 1 / |A|$ and $\varepsilon$-min-wise if
$\Pr( \min h(A) = x ) = (1 + \varepsilon) / |A|$. For this problem we show new tight bounds for $k=2,3,4$
of
$\varepsilon = \Theta(\sqrt{n}), \Theta(\log n), \Theta(\log n)$ respectively and for $k=5$ it is folklore
that $O(1)$-min-wise ($\varepsilon = \Oh(1)$) can be achieved.
Since tight bounds for $k \ge 5$ exist (see \Cref{kind:related})
, our contribution closes the problem.

{\bf Randomized quicksort.} Next we consider a classic sorting algorithm presented in many randomized
algorithms books, e.g. already on page three of Motwani-Raghavan~\cite{Motwani:1995:RA:211390}.
The classic analysis of quicksort in Motwani-Raghavan uses crucially the probability of a particular
element being mapped to the smallest hash value out of all the elements: the expected worst-case running
time in this analysis is $\Oh(n \log n \cdot \Pr( \min h(A) = x ))$, where $A$ is the set of $n$ elements
to be sorted and $x \in A$. It follows directly from our new tight min-wise bounds that this analysis
cannot be improved further. A special version of randomized quicksort was showed by Karloff and Raghavan
to use expected worst-case time $\mathcal{O}(n \log n)$ when the pivot elements are chosen using a $5$-independent
hash function~\cite{Karloff:1988:RAP:62212.62242}. Our main result is a new general bound for the number
of comparisons performed under limited independence, which applies to several settings of quicksort, including
the setting of Karloff-Raghavan where we show the same running time using only $4$-independence.
Furthermore,
we
show that $k=2$ and $k=3$ can imply expected worst-case time
$\Omega\left (n \log^2 n \right)$.
An interesting
observation is that our new bounds for $k=4$ and $k=2$ shows that the classic analysis using min-wise hashing
is not tight, as we go below those bounds by a factor $\log n$ for $k=4$ and a factor $\sqrt{n} / \log n$ for $k=2$.
Our findings imply that a faster $4$-independent hash function can be used to guarantee the optimal running
time for randomized quicksort, which could potentially be of practical interest. Interestingly, our new bounds
on the number of performed comparisons under limited independence has implications on classic algorithms for
binary planar partitions and treaps. For binary planar partitions our results imply expected partition size
$\Oh ( n \log n)$ for the classic randomized algorithm for computing binary planar
partitions~\cite[Page 10]{Motwani:1995:RA:211390} under $4$-independence. For randomized
treaps~\cite[Page 201]{Motwani:1995:RA:211390} our new results imply $\Oh (\log n)$ worst-case depth
for $4$-independence. 

{\bf Larget bucket size.}
The last setting we consider is throwing $n$ balls into $n$ buckets using an $k$-independent hash function
and analyzing the size of the largest bucket. This can be
regarded
as a load balancing as the balls can represent
``tasks'' and the buckets represent processing units. Our main result is a family of $k$-independent hash
functions, which when used in this setting implies largest bucket size $\Omega(n^{1/k})$ with constant
probability. This result was previously known only for $k=2$ due to Alon et al.
\cite{Alon:1997:LHG:258533.258639} and our result is a generalization of their bound.
As an example of the usefulness of such bucket size bounds, consider the fundamental data structure;
the dictionary. Widely used algorithms books such as Cormen et al.~\cite{Cormen:2001:IA:580470} teaches
as the standard method to implement a dictionary to use an array with \emph{chaining}. Chaining here simply
means that for each key, corresponding to an entry in the array, we have a linked list (chain) and when a new
key-value pair is inserted, it is inserted at the end of the linked list. Clearly then, searching for a
particular key-value pair takes worst-case time proportional to the size of the largest chain. Hence, if one
is interested in worst-case lookup time guarantees then the expected largest bucket size formed by the keys
in the dictionary is of great importance.

\section{Relation to previous work}\label{kind:related}
We will briefly review related work on the topic of bounding the independence used as well as mention some of the popular hash function constructions.

The line of research that considers the amount of independence required is
substantial.
As examples,
Pagh et al.~\cite{Pagh:2007:LPC:1250790.1250839} showed that linear probing works with $5$-independence.
For the case of $\varepsilon$-min-wise hashing (``almost'' min-wise-hashing as used e.g. in
\cite{Indyk:1999:SAM:314500.314600}) Indyk showed that
$\mathcal{O}(\log \frac{1}{\varepsilon})$-independence is sufficient. For both of the above problems
Thorup and P\v{a}tra\c{s}cu~\cite{thorupind} showed optimality: They show existence of explicit families
of hash functions that for linear probing is $4$-independent leading to $\Omega(\log n)$ probes and
for $\varepsilon$-min-wise hashing is $\Omega(\log \frac{1}{\varepsilon})$-independent that implies
$(2\varepsilon )$-min-wise hashing. Additionally, they show that the popular multiply-shift hashing
scheme by Dietzfelbinger et al.~\cite{Dietzfelbinger199719} is not sufficient for linear probing and
$\varepsilon$-min-wise hashing. In terms of lower bounds, it was shown by
Broder et al.\cite{Broder98min-wiseindependent} that $k=2$ implies $\Pr( \min h(A) = x ) = 1 / \sqrt{|A|}$.
We provide a matching lower bound and new tight bounds for $k=3,4$. Additionally we review a folklore
$\Oh (1 / n)$ upper bound for $k=5$. 
Our lower bound proofs for min-wise hashing (see \Cref{tab:minwise}) for $k=3,4$ are surprisingly
similar to those of Thorup and P\v{a}tra\c{s}cu for
linear probing, in fact we use the same ``bad'' families of hash functions but with a different analysis.
Further the same families imply the same multiplicative factors relative to the optimal. Our new tight
bounds together with the bounds for $k \geq 5$ due to \cite{Indyk:1999:SAM:314500.314600,thorupind}
provide the full picture of how min-wise hashing behaves under limited independence. 

Randomized quicksort\cite{Motwani:1995:RA:211390} is well known to sort $n$ elements in expected time
$\Oh ( n \log n)$ under full independence. Given that pivot elements are picked by having $n$ random
variables with outcomes ${0,\ldots,n-1}$ and the outcome of variable $i$ in the sequence determines
the $i$th pivot element, then running time $\Oh (n \log n)$ has been shown\cite{Karloff:1988:RAP:62212.62242}
for $k=5$. We improve this and show $\Oh (n \log n)$ time for $k=4$ in the same setting. To the knowledge of the authors,
 it is still an open problem to analyze the version of randomized quicksort under
limited independence as presented by e.g. Motwani-Raghavan. The analysis of both the randomized binary
planar partition algorithm and the randomized treap in Motwani-Raghavan is done using the exact same argument
as for quicksort, namely using min-wise hashing which we show cannot be improved further and is not tight.
Our new quicksort bounds directly translates to improvements for these two applications. The randomized
binary planar partition algorithm is hence improved to be of expected size $\Oh (n \log^2 n)$ for $k=2$
and $\Oh(n \log n)$ for $k=4$, and the
expected
worst case depth of any node in a randomized treap is improved to
be $\Oh(\log^2 n)$ for $k=2$ and $\Oh( \log n)$ for $k=4$. 

As briefly mentioned earlier, our largest bucket size result is related to the generalization of
Alon et al., STOC'97, specifically \cite[Theorem 2]{Alon:1997:LHG:258533.258639}. They show that for a
(perfect square) field $\F$ then the class $\mathcal{H}$ of all linear transformations between $\F^2$ and
$\F$ has the property that when a hash function is picked uniformly at random from $h \in \mathcal{H}$
then an input set of size $n$ exists so that the
largest bucket has size at least $\sqrt{n}$.
In terms of upper bounds for largest bucket size, remember that a family $\mathcal{H}_u$ of hash functions
that map from $\mathcal{U}$ to $[n]$ is \emph{universal}~\cite{Carter1979143} if for a $h$ picked uniformly
from $\mathcal{H}_u$
it holds
\[ \forall x \neq y \in \mathcal{U}: \Pr(h(x) = h(y)) \leq 1/n\text{.}\]
Universal hash functions are known to have expected largest bucket size at most $\sqrt{n} + 1/2$, hence
essentially tight compared to the bound $\sqrt{n}$ lower bound of Alon et al. On the other end of the
spectrum, full independence is known to give expected largest bucket size $\Theta( \log n / \loglog n)$
due to a standard application of Chernoff bounds. This bound was proven to hold for
$\Theta( \log n / \loglog n)$-independence as well~\cite{Schmidt:1995:CBA:217737.217746}. In
\Cref{sec:load:upper} we additionally review a folklore upper bound
coinciding
with our new $\Omega(n^{1/k})$ lower bound.

Since the question of how much independence is needed from a practical perspective often could be rephrased ``how fast a hash function can I use and maintain algorithmic guarantees?'' we will briefly recap some used hash functions and pseudorandom generators. Functions with lower independence are typically faster in practice than functions with higher. The formalization of this is due to Siegel's lower bound \cite{siegel2004} where he shows that in the cell probe model, to achieve $k$-independence and number of probes $t < k$ then you need space $k (n/k)^{1/t}$. Since space usage scales with the independence $k$ then for high $k$ the effects of the memory hierarchy will mean that even if the time is held constant the practical time will scale with $k$ as cache effects etc. impact the running time.

The most used hashing scheme in practice is, as mentioned, the $2$-independent multiply-shift by Dietzfelbinger et al.~\cite{Dietzfelbinger199719}, which can be twice as fast~\cite{Thorup:2000:ESU:338219.338597} compared to even the simplest linear transformation $x \mapsto (ax+b) \mod p$. For $3$-independence we have due to (analysis by) Thorup and P\v{a}tra\c{s}cu the simple tabulation scheme~\cite{Patrascu:2011:PST:1993636.1993638}, which can be altered to give $5$-universality~\cite{thorup:tabulation}. For general $k$-independent hash functions the standard solution is degree $k-1$ polynomials, however especially for low $k$ these are known to run slowly, e.g. for $k=5$ then polynomial hashing is $5$ times slower than the tabulation based solution of \cite{thorup:tabulation}. Alternatively for high independence the double tabulation scheme by Thorup\cite{thorup2013}, which builds on Siegels result \cite{siegel2004}, can potentially be practical. On smaller universes Thorup gives explicit and practical parameters for $100$-independence. Also for high independence, the nearly optimal hash function of Christiani et al.\cite{christiani2015} should be practical.
For generating $k$-independent variables then Christiani and Pagh's constant time generator~\cite{6979004} performs well - their method is at an order of magnitude faster than evaluating a polynomial using fast fourier transform. We note that even though constant time generators as the above exist, the practical evaluation will actually scale with the independence, as the memory usage of the methods depend on the independence and so the effects of the underlying memory hierarchy comes to effect.

Finally, we would like to note that the paradigm of independence has its limitations in the sense that even though one can prove that $k$-independence by itself does not imply certain algorithmic guarantees, it can not be ruled out that $k$-independent hash functions exist that do. That is, lower bound proofs typically construct artificial families to provide counter examples, which in practice would not come into play. As an example, consider that linear probing needs $5$-independence to work as mentioned above but it has been proven to work with simple tabulation hashing~\cite{Patrascu:2011:PST:1993636.1993638}, which only has $3$-independence.

\section{Our results}
With regard to min-wise hashing, we close this version of the problem by providing new and tight bounds for $k=2,3,4$. We consider the following setting: let $A$ bet a set of size $n$ and let $\mathcal{H}$ be a $k$-independent family of hash functions. We examine the probability of any element $x \in A$ receiving the smallest hash value $h(x)$ out of all elements in $A$ when $h \in \mathcal{H}$ is picked uniformly at random.
For the case of $k=2,3,4$-independent families we provide new bounds as shown in \Cref{tab:minwise}, which provides a full understanding of the parameter space as a tight bound of $\Pr( \min h(A) = x ) = \mathcal{O}(1/n)$ is known for $k\geq 5$ due to Indyk\cite{Indyk:1999:SAM:314500.314600}.
\begin{table*}
 \centering
\begin{tabular}{l | c c c c}
\hline\hline  & $k=2$ & $k=3$ & $k=4$ & $k \geq 5$
\\ [0.5ex]
\hline
\textsc{Upper bound} & $\mathcal{O}(\sqrt{n} / n)$  & $\mathcal{O}((\log n)/n)^*$ & $\mathcal{O}((\log n)/n)^*$ & $\mathcal{O}(1/n)$\\
\textsc{Lower bound} & $\Omega(\sqrt{n} / n)^*$ & $\Omega((\log n)/n)^*$ & $\Omega((\log n)/n)^*$ & $\Omega(1/n)$\\
[0.5ex] \hline
\end{tabular}
\vspace{0.5cm}
\caption{Result overview for min-wise hashing. Results in this paper are marked with $^*$. For a set $A$ of size $n$ and an element $x \in A$ the cells correspond the probability $\Pr( \min h(A) = x )$ for a hash function $h$ picked uniformly at random from a $k$-independent family $\mathcal{H}$. }
\label{tab:minwise}
\end{table*}
\noindent We make note that our lower bound proofs, which work by providing explicit ``bad'' families of functions, share similarity with Thorup and P\v{a}tra\c{s}cu's~\cite[Table 1]{thorupind} proof of linear probing. In fact, our bad families of functions used are exactly the same, while the analysis is different. Surprisingly, the constructions imply the same factor relative to optimal as in linear probing, for every examined value of $k$.

Next, we consider randomized quicksort under limited independence. In the same setting as Karloff and Raghavan~\cite{Karloff:1988:RAP:62212.62242} our main result is that $4$-independence is sufficient for the optimal $\mathcal{O}( n \log n)$ expected worst-case running time. The setting is essentially that pivot elements are picked from a sequence of $k$-independent random variables that are pre-computed. Our results apply to a related setting of quicksort as well as to the analysis of binary planar partitions and randomized treaps. Our results are summarized in \Cref{tab:quicksort}.

\begin{table*}
 \centering
\begin{tabular}{l | c c c c}
\hline\hline  & $k=2$ & $k=3$ & $k=4$ & $k \geq 5$
\\ [0.5ex]
\hline
\textsc{Upper bound} & $\mathcal{O}( n \log^2 n)^*$  & $\mathcal{O}( n \log^2 n)^*$ & $\mathcal{O}(n \log n)^*$ & $\mathcal{O}(n \log n)$\\
\textsc{Lower bound} & $\Omega(n \log n)$ & $\Omega(n \log  n)$ & $\Omega(n \log n)$ & $\Omega(n \log n)$\\
[0.5ex] \hline
\end{tabular}
\vspace{0.5cm}
\caption{Result overview for randomized quicksort. Results in this paper are marked with $^*$. When our hash function $h$ is picked uniformly from $k$-independent family $\mathcal{H}$ then the cells in the table denote the expected running time to sort $n$ distinct elements. The $5$-independent upper bound is from Karloff-Raghavan\cite{Karloff:1988:RAP:62212.62242}.}
\label{tab:quicksort}
\end{table*}

Finally for the fundamental case of throwing $n$ balls into $n$ buckets. The main result is a simple $k$-independent family of functions which when used to throw the balls imply that with constant probability the largest bucket has $\Omega(n^{1/k})$ balls. We show the theorem below.

\begin{theorem}\label{thm:loadmain}
	Consider the setting where $n$ balls are distributed among $n$ buckets using a
	random hash function $h$.
	For $m \le n$ and any $k \in \mathbb{N}$ such that $k < n^{1/k}$ and $m^k \ge n$ 
	a $k$-independent distribution over hash functions exists such that the largest bucket size
	is $\Omega(m)$ with probability $\Omega \left ( \frac{n}{m^k} \right )$ when $h$ is chosen
	according to this distribution.
\end{theorem}

An implication of \Cref{thm:loadmain} is that we now have the full understanding of the parameter space for this problem, as it was well known that independence $k = \mathcal{O}(\log n / \loglog n)$ implied $\Theta( \log n / \loglog n)$ balls in the largest bucket. We summarize with the corollary below.
\begin{corollary} \label{cor:full}
	Consider the setting where $n$ balls are distributed among $n$ buckets using a random hash function $h$.
	Given an integer $k$ a distribution over hash functions exists such that if $h$ is chosen according
	to this distribution then with $L$ being the size of the largest bucket
	\begin{enumerate}[label=(\alph*)]
	\item if $k \leq n^{1/k}$ then $L = \Omega \left( n^{1/k} \right)$ with probability
		$\Omega(1)$.
	\item if $k > n^{1/k}$ then $L = \Omega \left( \log n / \log \log n \right )$ with probability
		$\Omega(1)$.
	\end{enumerate}
\end{corollary}
We note that the result of \Cref{thm:loadmain} is not quite the generalization of the lower bound of Alon et al. since they show $\Omega(n^{1/2})$ largest bucket size for any linear transformation while our result provides an explicit worst-case $k$-independent scheme to achieve largest bucket size $\Omega( n^{1/k})$. However, as is evident from the proof in the next section, our scheme is not that artificial: In fact it is ``nearly'' standard polynomial hashing, providing hope that the true generalization of Alon et al. can be shown.

\section{Preliminaries}
We will introduce some notation and fundamentals used in the paper. For an integer $n$ we let $[n]$ denote $\{0, \ldots, n-1\}$. For an event $E$ we let $[E]$ be the variable that is $1$ if $E$ occurs and $0$ otherwise. Unless explicitly stated otherwise, the $\log n$ refers to the base $2$ logarithm of $n$. For a real number $x$ and a non-negative integer $k$ we define $x^{\underline{k}}$ as $x(x-1)\ldots(x-(k-1))$.

The paper is about application bounds when the independence of the random variables used is limited. We define independence of a hash function formally below. 
\begin{definition}\label{def:hashind}
Let $h: \mathcal{U} \mapsto V$ be a random hash function, $k \in \N$ and let $u_1, \ldots, u_k$ be any distinct $k$ elements from $\mathcal{U}$ and $v_1, \ldots, v_k$ be any $k$ elements from $V$.\\
Then $h$ is $k$-independent if it holds that
\[ \Pr \left( h(u_1) = v_1 \wedge \ldots \land h(u_k) = v_k \right) = \frac{1}{|V|^k}. \]
\end{definition}
Note that an equivalent definition for a sequence of random variables hold: they are $k$-independent if any element
is uniformly distributed and every $k$-tuple of them is independent.

\section{Min-wise hashing}

In this section we show the bounds that can be seen in \Cref{tab:minwise}.
As mentioned earlier, there is a close relationship between the worst case query time of an element in
linear probing and min-wise hashing when analysed under the assumption of hash functions with limited independence.
Intuitively, long query time for linear probing is caused by many hash values being ``close'' to the hash value
of the query element. On the other hand a hash value is likely to be the minimum if it is ``far away'' from the
other hash values. So intuitvely, min-wise hashing and linear probing are related by the fact that
good guarantees require a ``sharp'' concentration on how close to the hash value of the query element the
other hash values are.

\subsection{Upper bounds}
We show the following theorem which results in the upper bounds shown in \Cref{tab:minwise}. Note that the bound for $4$-independence follows trivially from the bound for $3$-independence and that the $5$-independence bound is folklore but included for completeness.
\begin{theorem}
Let $X = \set{x_0,x_1,\ldots,x_n}$ and $h : X \to (0,1)$ be a hash function. If $h$ is $3$-independent
\[
	\Pr \left ( h(x_0) < \min_{i \in \set{1,\ldots,n}} \set{h(x_i)} \right )
	=
	\Oh \left ( \frac{\log (n+1)}{n+1} \right )
\]
If $h$ is $5$-independent
\[
	\Pr \left ( h(x_0) < \min_{i \in \set{1,\ldots,n}} \set{h(x_i)} \right )
	=
	\Oh \left ( \frac{1}{n+1} \right )
\]
\end{theorem}
\begin{proof}
For notational convenience let $E$ denote the event
$\left ( h(x_0) < \min_{i \in \set{1,\ldots,n}} \set{h(x_i)} \right )$.
First assume that $h$ is $3$-independent. Fix $h(x_0) = \alpha \in (0,1)$. Then $h$ is $2$-independent
on the remaining keys. Let $Z = \sum_{i = 1}^n \left [h(x_1) \le \alpha \right ]$. Then
under the assumption $h(x_0) = \alpha$:
\[
	\Pr \left ( E \mid h(x_0) = \alpha \right )
	=
	\Pr \left ( Z = 0  \mid h(x_0) = \alpha \right )
	\le
	\Pr \left ( \abs{Z - \E Z} \ge \E Z  \mid h(x_0) = \alpha  \right )
\]
Now since $h$ is $2$-independent on the remaining keys we see that $\Pr \left ( E \mid h(x_0) = \alpha \right )$
is upper bounded by (using \Cref{fact:kthmoment}):
\begin{align}
	\Pr \left ( \abs{Z - \E Z} \ge \E Z  \mid h(x_0) = \alpha  \right )
	&\le
	\frac{\E \left ( \left ( Z - \E Z \right )^2 \right )}{\left (\E Z \right )^2}
	=
	\Oh \left (
		\frac{1}{\E Z}
	\right )\notag\\
	&=
	\Oh \left (
		\frac{1}{n\alpha}\label{eq:proofForK3}
	\right )
\end{align}
Hence:
\begin{align}
	\label{eq:conclusionK3}
		\Pr \left ( E \mid h(x_0) = \alpha \right )
	&
		=
		\int_0^1
			\Pr \left ( E \mid h(x_0) = \alpha \right )
			d\alpha
	\notag
	\\
	&
		\le
		\frac{1}{n} +
		\int_{1/n}^1
		\Oh \left (
			\frac{1}{n\alpha}
		\right )
		=
		\Oh \left ( \frac{\log (n+1)}{n+1} \right )
\end{align}

This proves the first part of the theorem. Now assume that $h$ is $5$-independent and define $Z$
in the same way as before. In the same manner as we established the upper bound for
$\Pr \left ( E \mid h(x_0) = \alpha \right )$ in \Cref{eq:proofForK3} we see that it is now upper bounded by
(using \Cref{fact:kthmoment}):
\begin{align*}
	\Pr \left ( \abs{Z - \E Z} \ge \E Z  \mid h(x_0) = \alpha  \right )
	&\le
	\frac{\E \left ( \left ( Z - \E Z \right )^4 \right )}{\left (\E Z \right )^4}\\
	&=
	\Oh \left (
		\frac{1}{\left ( \E Z \right )^2}
	\right )
	=
	\Oh \left (
		\frac{1}{\left (n\alpha \right )^2}
	\right )
\end{align*}
In the same manner as in \Cref{eq:conclusionK3} we now see that
\[
	\Pr \left ( E \right )
	=
	\int_0^1
		\Pr \left ( E \mid h(x_0) = \alpha \right )
		d\alpha
	\le
	\frac{1}{n} +
	\int_{1/n}^1
	\Oh \left (
		\frac{1}{(n\alpha)^2}
	\right )
	=
	\Oh \left ( \frac{1}{n+1} \right )
\]
\qed
\end{proof}

\subsection{Lower bounds}
We first show the $k=4$ lower bound seen in \Cref{tab:minwise}. As mentioned earlier, the argument follows from the same ``bad'' distrubition as Thorup and P\v{a}tra\c{s}cu\cite{thorupind}, but with a different analysis.
\begin{theorem}
	For any key set $X = \set{x_0,x_1,\ldots,x_n}$ there exists a random hash function
	$h : X \to (0,1)$ that is $4$-independent such that
	\begin{align}
		\label{eq:lbminwise4}
		\Pr \left (
			h(x_0) < \min \set{h(x_1), \ldots, h(x_n)}
		\right )
		=
		\Omega \left ( \frac{\log (n+1)}{n+1} \right )
	\end{align}
\end{theorem}
\begin{proof}
	We consider the strategy from Thorup and P\v{a}tra\c{s}cu~\cite[Section 2.3]{thorupind} where we hash
	$X$ into $[t]$, where $t$ power of $2$ such that $t = \Theta(n)$.
	We use the strategy to determine the first $\log t$ bits of the values of $h$ and let the remaining bits
	be chosen independently and uniformly at random. The strategy ensures that for every
	$\ell \in \left [ \frac{2}{3} \log t, \frac{5}{6} \log t \right ]$ with probability $\Theta(2^\ell/n)$
	there exists an interval $I$ of size $\Theta(2^{-\ell})$ such that $h(x_0)$ is uniformly distributed
	in $I$ and $I$ contains at least $\frac{t}{2^\ell} \cdot (1 + \Omega(1))$ keys from $X$. Furthermore
	these events are disjoint.
	From the definition of the algorithm we see that for every
	$\ell \in \left [ \frac{2}{3} \log t, \frac{5}{6} \log t \right ]$ with probability $\Theta(2^\ell/n)$
	there exists an interval $I$ of size $\Theta(2^{-\ell})$ such that $h(x_0)$ is uniformly distributed
	in $I$ and $I$ contains no other element than $h(x_0)$.
	Let $y$ be the maximal value of all of $h(x_1), \ldots, h(x_n)$ which are smaller than $h(x_0)$ and
	$0$ if all hash values are greater than $h(x_0)$. Then we know that:
	\[
		\E ( h(x_0) - y ) \ge
		\sum_{\ell \in \left [ \frac{2}{3} \log t, \frac{5}{6} \log t \right ]}
		\Theta \left ( \frac{2^\ell}{n} \right ) \cdot \Theta(2^{-\ell})
		= \Theta \left ( \frac{\log n}{n}  \right )
	\]
	We know define the hash function $h' : X \to (0,1)$ by $h'(x) = \left ( h(x) - z \right )\bmod 1$
	where $z \in (0,1)$ is chosen uniformly at random. Now fix the choice of $h$. Then $h'(x_0)$ is smaller
	than $\min \set{h'(x_1), \ldots, h'(x_n)}$ if $z \in (y, h(x_0))$. Hence for this fixed choice of $h$:
	\[
		\Pr \left (
			h'(x_0) < \min \set{h'(x_1), \ldots, h'(x_n)}
			\mid
			h
		\right )
		\ge
		h(x_0) - y
	\]
	Therefore
	\begin{align*}
		\Pr \left (
			h'(x_0) < \min \set{h'(x_1), \ldots, h'(x_n)}
		\right )
		&\ge
		\E \left ( h(x_0) - y\right )\\
		&=
		\Omega \left ( \frac{\log n}{n} \right )
		=
		\Omega \left ( \frac{\log (n+1)}{n+1} \right )
	\end{align*}
	and $h$ satisfies \Cref{eq:lbminwise4}
	\qed
\end{proof}

The lower bound for $k=2$ seen in \Cref{tab:minwise} is shown in the following theorem, using a probabilistic mix between distribution strategies as the main ingredient.
\begin{theorem}
	For any key set $X = \set{x_0,x_1,\ldots,x_n}$ there exists a random hash function
	$h : X \to [0,1)$ that is $2$-independent such that
	\[
		\Pr \left ( h(x_0) < \min_{i \in \set{1,\ldots,n}} \set{h(x_i)} \right )
		=
		\Omega \left ( \frac{1}{\sqrt{n}} \right )
	\]
\end{theorem}
\begin{proof}
Since we are only interested in proving the asymptotic result, and have no
intentions of optimizing the constant we can wlog. assume that $10\sqrt{n}$
is an integer that divides $n$. To shorten notation we let $\ell = 10\sqrt{n}$.

We will now consider four different strategies for assigning $h$, and they
will choose a hash function $g : X \to [\ell+1]$. Then we let $(U_x)_{x \in X}$
be a family of independent random variables uniformly distributed in $(0,1)$ and
define $h(x) = \frac{g(x) + U_x}{\ell+1}$. The high-level approach is to define distribution strategies such that
some have too high pair-collision probability, some have too low and likewise for the probability of hashing to the same value as $x_0$. Then we mix over the strategies
with probabilities such that in expectation we get the correct number of collisions but we maintain and increased probability of $x_0$ hashing to a smaller value than the rest of the keys.
We will now describe the four strategies for choosing $g$.
\begin{itemize}
\item \textbf{Strategy $S_1$:} $g(x_0)$ is uniformly chosen.
Then $(g(x))_{x \neq x_0}$ is chosen uniformly at random such that $g(x) \neq g(x_0)$
and for each $y \neq g(x_0)$ there are exactly $\frac{n}{\ell}$ hash values
equal to $y$.

\item \textbf{Strategy $S_2$:} $g(x_0)$ is uniformly chosen, and $y_1$ is uniformly
chosen such that $y_1 \neq g(x_0)$. For each $x \in X \backslash \set{x_0}$
we define $g(x) = y_1$.

\item \textbf{Strategy $S_3$:} $g(x_0)$ is uniformly chosen.
Then $Z \subset X$ is chosen uniformly at random such that $\abs{Z} = \frac{\sqrt{n}}{5}$.
We define $g(z) = g(x_0)$ for every $z \in Z$.
Then $(g(x))_{x \neq x_0, x \notin Z}$ is chosen uniformly at random under the
constraint that $g(x) \neq g(x_0)$ and for each $y \neq g(x_0)$ there are at most
$\frac{n}{\ell}$ hash values equal to $y$.

\item \textbf{Strategy $S_4$:} $y \in [\ell+1]$ is uniformly chosen and $g(x) = y$
for each $x \in X$.
\end{itemize}
For each of the four strategies we compute the probability that $g(x_0) = g(x)$
and $g(x) = g(x')$ for each $x,x' \in X \backslash \set{x_0}$. Because of
symmetry the answer is independent of the choice of $x$ and $x'$. This is
a trivial exercise and the results are summarized in \cref{probabilities}.
\begin{table}
    \centering
	\begin{tabular}{|c|c|c|}
    \hline
		Strategy &
		$\Pr_{S_i} \left ( g(x_0) = g(x)  \right )$ &
		$\Pr_{S_i} \left ( g(x)   = g(x') \right )$ \\
        [0.5ex]
        \hline
		{\bf $S_1$} &
		$0$ &
		$\frac{\frac{n}{\ell}-1}{n-1} \left ( < \frac{1}{\ell+1} \right )$ \\
		\hline
		{\bf $S_2$} &
		$0$ &
		$1$ \\
		\hline
		{\bf $S_3$} &
		$\frac{1}{5\sqrt{n}}$ &
		$\le \frac{\frac{n}{\ell}-1}{n-1} +
			\frac{\frac{\sqrt{n}}{5} \left ( \frac{\sqrt{n}}{5} - 1 \right )}{n(n-1)}
			\left ( < \frac{1}{\ell+1} \right )$
		\\
		\hline
		{\bf $S_4$} &
		$1$ &
		$1$ \\
		\hline
	\end{tabular}
	\caption{Strategies for choosing function $h$ and their collision probabilities for $x,x' \in X \backslash \set{x_0}$. The main idea is that there are two strategies with too low probability and two with too high probability, for both types of collisions. However, we can mix probabilistically over the strategies to achieve the theorem.}
	\label{probabilities}
\end{table}

\noindent For event $E$ and strategy $S$ let $\Pr_S(E)$ be the probability of $E$ under strategy $S$.
First we define the strategy $T_1$ that chooses strategy $S_1$ with probability
$p_1$ and strategy $S_2$ with probability $1 - p_1$. We choose $p_1$ such that
$\Pr_{T_1} \left ( g(x) = g(x') \right ) = \frac{1}{\ell+1}$. Then
$p_1 > 1 - \frac{1}{\ell+1}$.
Likewise we define the strategy $T_2$ that chooses strategy $S_3$ with probability
$p_2$ and strategy $S_4$ with probability $1 - p_2$ such that
$\Pr_{T_2}\left ( g(x) = g(x') \right ) = \frac{1}{\ell+1}$. Then $p_2 > 1 - \frac{1}{\ell+1}$
as well. Then:
\[
	 \Pr_{T_1} \left ( g(x) = g(x_0) \right )
	 = 0
	 <
	 \frac{1}{\ell+1}
	 <
	 \frac{2}{\ell}
	 =
	 \frac{1}{5\sqrt{n}}
	 \le
	 \Pr_{T_2} \left ( g(x) = g(x_0) \right )
\]
Now we define strategy $T^*$ that chooses strategy $T_1$ with probability $q$ and
$T_2$ with probability $1-q$. We choose $q$ such that
$\Pr_{T_2} \left ( g(x) = g(x_0) \right ) = \frac{1}{\ell+1}$. Then
$q \ge 1 - \frac{\frac{1}{\ell+1}}{\frac{2}{\ell}} \ge \frac{1}{2}$. Hence $T^*$
chooses strategy $S_1$ with probability $\ge \frac{1}{2} \left ( 1 - \frac{1}{\ell+1}\right)
= \Omega(1)$.

The strategy $T^*$ implies a $2$-independent $g$, since due to the the mix of strategies the pairs of keys collide with the correct probability, that is, the same probability as under full independence. Further, with constant probability $g(x_0)$
is unique. Hence with probability $\Omega \left ( \frac{1}{\ell+1} \right ) =
\Omega \left ( \frac{1}{\sqrt{n}} \right )$, $g(x_0) = 0$ and $g(x_0)$ is unique.
In this case $h(x_0)$ is the minimum of of all $h(x), x \in X$ which concludes the
proof.
\qed
\end{proof} 
\section{Quicksort}

The textbook version of the quicksort algorithm, as explained in~\cite{Motwani:1995:RA:211390}, is the following. As input we are given a set of $n$ numbers $S = \set{x_0,\ldots,x_{n-1}}$ and we uniformly at random choose a pivot element $x_i$. We then
compare each element in $S$ with $x_i$ and determine the sets $S_1$ and $S_2$ which consist of the elements that
are smaller and greater than $x_i$ respectively. Then we recursively call the procedure on $S_1$ and $S_2$ and output the sorted sequence $S_1$ followed by $x_i$ and $S_2$. For this setting there are to the knowledge of the authors no known bounds under limited independence.

We consider two different settings where our results seen in \Cref{tab:quicksort} apply.\\ 
{\bf Setting 1}. Firstly, we consider the same setting as in \cite{Karloff:1988:RAP:62212.62242}. Let the input again be $S = \set{x_0,\ldots,x_{n-1}}$. The pivot elements are pre-computed the following way: let random variables $Y_1, \ldots, Y_n$ be $k$-independent and each $Y_i$ is uniform over $[n]$. The $i$th pivot element is chosen to be $x_{Y_i}$. Note that the sequence of $Y_i$'s is not always a permutation, hence a ``cleanup'' phase is necessary afterwards in order to ensure pivots have been performed on all elements.\\
{\bf Setting 2}. The second setting we consider is the following. Let $Z = Z_1, \ldots, Z_n$ be a sequence of $k$-independent random variables that are uniform over the interval $(0,1)$. Let $\min(j,Z)$ denote the index $i$ of the $j$'th smallest $Z_i$. We choose pivot element number $j$ to be $x_{\min(j,Z)}$. Note that the sequence $Z$ here defines a permutation with high probability and so we can simply repeat the random experiments if any $Z_i$ collide. 

In this section we show the results of \Cref{tab:quicksort} in Setting 1. We refer to \Cref{sec:settingtwo}
for proofs for Setting 2 and note that the same bounds apply to both settings.

Recall, that we can use the results on min-wise hashing to show upper bounds on the running time. The key
to sharpening this analysis is to consider a problem related to that of min-wise hashing. In \Cref{lem:boundC4} 
we show that for two sets $A,B$ satisfying $\abs{A} \le \abs{B}$ there are only $O(1)$ pivot elements chosen
from $A$ before the first element is chosen from $B$. We could use a min-wise type of argument to show that a
single element $a \in A$ is chosen as a pivot element before the first pivot element is chosen from $B$ with
probability at most $\Oh \left ( \frac{\log n}{\abs{B}} \right )$. However, this would only gives us an upper bound
of $\Oh \left ( \log n \right )$ and not $\Oh ( 1 )$.

\begin{lemma}
	\label{lem:boundC4}
	Let $h : [n] \to [n]$ be a $4$-independent hash function and let $A, B \subset [n]$ be disjoint sets
	such that $\abs{A} \le \abs{B}$. Let $j \in [n]$ be the smallest value such that $h(j) \in B$, and
	$j = n$ if no such $j$ exist. Then let $C$ be the number of $i \in [j]$ such that $h(i) \in A$, i.e.
	\[
		C
		=
		\abs{
			\set{
				i \in [n]
				\mid
				h(i) \in A, \
				h(0), \ldots, h(i-1) \notin B
			}
		}
	\]
	Then $\E \left ( C \right ) = O(1)$.
\end{lemma}

Before we prove \Cref{lem:boundC4} we first show how to apply it to guarantee that 
quicksort only makes $\Oh \left ( n \log n \right )$ comparisons.

\begin{theorem}
	\label{thm:mainQS}
	Consider quicksort in Setting 1 where we sort a set $S = \set{x_0,\ldots,x_{n-1}}$ and pivot elements are chosen using a $4$-independent hash function. For any $i$
	the expected number of times $x_i$ is compared with another element $x_j \in S \backslash \set{x_i}$ 
	when $x_j$ is chosen as a pivot element is $\Oh \left ( \log n \right )$. In particular the expected running
	time is $\Oh \left ( n \log  n \right )$.
\end{theorem}
\begin{proof}
	Let $\pi : [n] \to [n]$ be a permutation of $[n]$ such that $x_{\pi(0)},\ldots,x_{\pi(n-1)}$ is sorted
	ascendingly. Then $\pi \circ h$ is a $k$-independent function as well, and therefore wlog. we
	assume that $x_0,\ldots,x_{n-1}$ is sorted ascendingly.
	
	Fix $i \in [n]$ and let $X = \set{x_{i+1}, \ldots, x_{n-1}}$. First we will upper bound the expected number of comparisons $x_i$ makes with elements from $X$ when an element of $X$ is chosen as pivot.
	We let $A_\ell$ and $B_\ell$ be the sets defined by
	\begin{align*}
		A_\ell = \set{x_j \mid j \in \left [ i, i+2^{\ell-1} \right ) \cap [n]}
		B_\ell = \set{x_j \mid j \in \left [ i+2^{\ell-1}, i+2^{\ell} \right ) \cap [n]}
	\end{align*}
	For any $x_j \in A_\ell$, $x_j$ is compared with $x_i$ only if it is chosen as a pivot element before 
	any element of $B_\ell$ is chosen as a pivot element. By \Cref{lem:boundC4} the expected number of times
	this happens is $\Oh(1)$ for a fixed $\ell$ since $\abs{B_\ell} \ge \abs{A_\ell}$. Since $A_\ell$ is empty
	when $\ell > 1 + \log n$ we see that $x_i$ is in ms{expectation only compared $\Oh \left ( \log n \right )$ times to the elements of $X$}. 
	We use an analogous argument to count the number of comparisons between $x_i$
	and $x_0,x_1,\ldots,x_{i-1}$ and so we have that every element makes in expectation $\Oh ( \log n)$ comparisons. As we have $n$ elements it follows directly from linearity of expectation that the total number of comparisons made is in expectation $\Oh ( n \log n)$. The last minor ingredient is the running time of the cleanup phase of Setting 1. We show in \Cref{lem:cleanup} that this uses expected time $\Oh( n \log n)$ for $k=2$, hence the stated running time of the theorem follows.
	\qed
\end{proof}

We now show \Cref{lem:boundC4}, which was a crucial ingredient in the above proof. 

\begin{proof}[of \Cref{lem:boundC4}]
Wlog.~assume that $\abs{A} = \abs{B}$ and let $m$ the size of $A$ and $B$. Let $\alpha = \frac{m}{n}$.

For each non-negative integer $\ell \ge 0$ let $C_\ell = \set{i \in [n] \mid i < 2^\ell \mid h(i) \in A}$.
Let $E_\ell$ be the event that $h(j) \notin B$ for all $j \in [n]$ such that $j < 2^\ell$. It is now
easy to see that if $i \in C$ then for some integer $\ell \le 1 + \lg n$, $i \in C_\ell$ and $E_{\ell-1}$ occurs.
Hence:
\begin{align}
	\label{eq:splitBound}
	\E(C) \le 
	\sum_{\ell = 0}^{\floor{\lg n}+1}
	\E \left ( \abs{C_\ell} \cdot \left [ E_{\ell-1} \right ] \right )
\end{align}
Now we note that
\begin{equation}\label{eq:expectationbound}
	\E \left ( \abs{C_\ell} \left [ E_{\ell-1} \right ] \right )
	\le
	\E \left (  \left ( \abs{C_\ell} - \alpha 2^{\ell+1} \right )^+ \right )
	+
	\E \left ( \alpha 2^{\ell+1} \cdot \left [ E_{\ell-1} \right ] \right )
\end{equation}
where $x^+$ is defined as $\max\set{x,0}$.

First we will bound $\E \left ( \left ( \abs{C_\ell} - \alpha 2^{\ell+1} \right )^+ \right )$ when
$\alpha 2^{\ell} \ge 1$.
Note that for any $r \in \mathbb{N}$:
\begin{align}\label{eq:term1}
	\Pr\left ( (\abs{C_\ell} - \alpha 2^{\ell+1})^+ \ge r \right )
	&=
	\Pr \left ( \abs{C_\ell} - \E(\abs{C_\ell}) \ge \alpha 2^{\ell} + r \right )\\\label{eq:4mom}
	&\le
	\frac{ \E \left ( \abs{C_\ell} - E \abs{C_\ell} \right )^4}{(\alpha 2^{\ell} + r)^4} 
\end{align}
Now consider \Cref{fact:kthmoment,fact:funkysum} which we will use together with \Cref{eq:term1}.
\begin{fact}\label[fact]{fact:kthmoment}
	Let $X = \sum_{i=1}^n X_i$ where $X_1,\ldots,X_i$ are $k$-independent random variables in $[0,1]$ for
	some even constant $k \ge 2$. Then
	\[
		\E \left ( \left ( X - \E X \right )^k \right ) = 
		\Oh \left ( 
			\left ( \E  X \right ) + \left ( \E  X \right )^{k/2}
		\right )
	\]
\end{fact}
\begin{fact}\label[fact]{fact:funkysum}
Let $r,l \in \R$. It holds that \[ \sum_{l \geq 1} \frac{1}{(r+l)^4} \leq \frac{1}{r^3} \text{.}\]
\end{fact}
\begin{proof}
We have \[ \sum_{l \geq 1} \frac{1}{(r+l)^4} \leq  \int_0^{\infty} \frac{1}{(r+x)^4}\mathrm{d}x = \left[ - \frac{1}{3} \frac{1}{(r+x)^3} \right]_0^\infty \leq  \frac{1}{r^3} \text{.}\]
\qed
\end{proof}
Note that whether each element $i \in [n], i < 2^k$ is lies in $C_\ell$ is only dependent on $h(i)$. Hence
$\abs{C_\ell} = \sum_{i \in [n], i < 2^k} [h(i) \in A]$ is the sum of $4$-independent variables with with 
values in $[0,1]$ and hence we can use \Cref{fact:kthmoment} to give an upper bound on \Cref{eq:term1}.
Combining \Cref{fact:kthmoment,fact:funkysum,eq:term1} we see that:
\begin{align}
	\label{eq:Cbound1}
		\E \left ( \left ( \abs{C_\ell} - \alpha 2^{\ell+1} \right )^+ \right )
	&
		= 
		\sum_{r \ge 1}
		\Pr\left ( (\abs{C_\ell} - \alpha 2^{\ell+1})^+ \ge r \right )
	\notag
	\\ 
	&
		\le 
		\sum_{r \ge 1}
		\frac{ \E \left ( \abs{C_\ell} - E \abs{C_\ell} \right )^4}{(\alpha 2^{\ell} + r)^4}
	\notag
	\\
	&
		= 
		\Oh \left (
			\frac{\left ( \alpha 2^{\ell} \right )^2}{\left ( \alpha 2^{\ell} \right )^3}
		\right )
		=
		\Oh \left (
			\frac{1}{\alpha 2^{\ell}}
		\right )
\end{align}

We we will bound $\E \left ( \alpha 2^{\ell+1} \cdot \left [ E_{\ell-1} \right ] \right )$
(the second term of \Cref{eq:expectationbound}) in a similar fashion still assuming that $\alpha 2^\ell \ge 1$.
For each $i \in [n]$ such that $i < 2^{\ell-1}$ let $Z_i = 1$ if $h(i) \in B$ and $Z_i = 0$ otherwise.
Let $Z$ be the sum of these $4$-independent variables, then $E_k$ is equivalent to $Z = 0$.
By \Cref{fact:kthmoment}
\[
	\E \left ( \left [ E_{\ell-1} \right ]  \right )
	=
	\Pr \left (
		Z = 0
	\right )
	\le
	\Pr \left (
		\abs{Z - \E Z} \ge \E Z
	\right )
	\le
	\frac{E (Z - \E Z)^4}{(\E Z)^4}
	=
	\Oh \left ( \frac{1}{(\E Z)^2} \right )
\]
Since $\E(Z) = \alpha \ceil{2^{\ell-1}}$ we see that
\begin{align}
	\label{eq:Cbound2}
	\alpha 2^{\ell+1} \cdot  \E \left ( \left [ E_k \right ] \right )
	=
	\Oh \left (
		\frac{1}{\alpha 2^\ell}
	\right )
\end{align}
By combining \Cref{eq:expectationbound,eq:Cbound1,eq:Cbound2} we see that for any $\ell$ such that $\alpha 2^\ell \ge 1$:
\begin{align}
	\label{eq:allBigEll}
	\E \left ( \abs{C_\ell} \left [ E_{\ell-1} \right ] \right )
	\le 
	\Oh \left (
		\frac{1}{\alpha 2^\ell}
	\right )
\end{align}
Furthermore, for any $\ell$ such that $\alpha 2^\ell \le 1$ we trivially get:
\begin{align}
	\label{eq:allSmallEll}
	\E \left ( \abs{C_\ell} \left [ E_{\ell-1} \right ] \right )
	\le 
	\E \left ( \abs{C_\ell} \right )
	\le 
	2^\ell \alpha
\end{align}
To conclude we combine \Cref{eq:splitBound,eq:allBigEll,eq:allSmallEll} and finish the proof
\[
	\E \left ( C \right )
	\le
	\Oh \left ( \sum_{\ell, \alpha 2^\ell \ge 1} \frac{1}{\alpha 2^\ell} \right )
	+
	\Oh \left ( \sum_{\ell, \alpha 2^\ell \le 1} \alpha 2^\ell \right )
	=
	\Oh(1)
\]
\qed
\end{proof}
We now show that the cleanup phase as described by Setting 1 takes $\Oh(n \log n)$ for $k=2$, which
means it makes no difference to asymptotic running time of quicksort. 
\begin{lemma}\label{lem:cleanup}
	Consider quicksort in Setting 1 where we sort a set $S = \set{x_0,\ldots,x_{n-1}}$ with a 
	$2$-independent hash function. The cleanup phase takes $\Oh \left ( n \log n \right )$ time.
\end{lemma}
\begin{proof}
	Assume wlog.~that $n$ is a power of $2$. For each $\ell \in \set{0,1,\ldots,\lg n}$ let $A_\ell$
	be the set of dyadic intervals of size $2^\ell$, i.e.
	\[
		A_\ell = \set{ \left [ i 2^\ell, (i+1)2^\ell \right ) \cap [n] \mid i \in \left [ n 2^{-\ell} \right ] }
	\]
	For any consecutive list of $s$ elements $x_i,\ldots,x_{i+s-1}$ such that none of them are chosen
	as pivot elements, there exist a dyadic interval $I$ of size $\Omega(s)$ such that none of $x_j, j \in I$
	are chosen as pivot elements. Hence we only need to consider the time it takes to sort elements corresponding
	to dyadic intervals. Let $P_\ell$ be an upper bound on the probability that no element from $\left [ 0, 2^\ell \right )$
	is chosen as a pivot element. Then the total running time of the cleanup phase is bounded by:
	\begin{align}
		\label{eq:fissefissefisse}
		\Oh \left (
			\sum_{\ell = 0}^{\lg n} 
				\abs{A_\ell} P_\ell 2^{2\ell}
		\right)
		=
		\Oh \left (
			n
			\sum_{\ell = 0}^{\lg n} 
				2^{\ell} \P_\ell 
		\right)
	\end{align}
	Fix $\ell$ and let $X = \sum_{i=0}^{n-1} \left [h(i) \in \left [ 0, 2^\ell \right ) \right ]$. Then by
	$\E(X) = 2^\ell$, so by Markov's inequality
	\begin{align*}
		\Pr \left (
			X = 0
		\right )
		&\le 
		\Pr \left (
			(X - \E(X))^2 \ge \left(\E(X)\right)^2
		\right )\\
		&\le 
		\frac{\E\left((X - \E(X))^2\right)}{(\E(X))^2}
		=
		\Oh \left (
			\frac{1}{\E(X)}
		\right )
		=
		\Oh \left ( 2^{-\ell} \right )
	\end{align*}
	Plugging this into \Cref{eq:fissefissefisse} shows that the running time is bounded by $\Oh \left ( n \log n \right )$. \qed
\end{proof}
Finally we show the new $2$-independent bound. The argument follows as the $4$-independent argument, except with $2$nd moment bounds instead of $4$th moment bounds.
\begin{theorem}
	\label{thm:2QS}
	Consider quicksort in Setting 1 where we sort a set $S = \set{x_0,\ldots,x_{n-1}}$ and pivot elements are chosen using a $2$-independent hash function. For any $i$
	the expected number of times $x_i$ is compared with another element $x_j \in S \backslash \set{x_i}$ 
	when $x_j$ is chosen as a pivot element is $\Oh \left ( \log^2 n \right )$. In particular the expected running
	time is $\Oh \left ( n \log^2  n \right )$.
\end{theorem}
\begin{proof}
The proof for $\Oh(n \log^2 n)$ expected running time follows from an analogous argument as \Cref{thm:mainQS}. The main difference being that the analogous lemma to \Cref{lem:boundC4} yields  $\E(C) = \Oh( \log n)$ instead of $\E(C) = \Oh(1)$, which implies the stated running time. This is due to the fact that as we have $2$-independence we must use the weaker $2$nd moment bounds instead of $4$th moment bounds as used e.g. in \Cref{eq:4mom}. Since the cleanup phase takes time $\Oh( n \log n)$ time even for $k=2$ due to \Cref{lem:cleanup} the stated time holds. Otherwise the proof follows analogously and we omit the full argument due to repetetiveness. \qed
\end{proof}
\subsection{Binary planar partitions and randomized treaps}
The result for quicksort shown in \Cref{thm:mainQS} has direct implications for two classic randomized algorithms. Both algorithms are explained in common text books, e.g. Motwani-Raghavan. 

A straightforward analysis of randomized algorithm\cite[Page 12]{Motwani:1995:RA:211390} for construction binary planar bipartitions simply uses min-wise hashing to analyze the expected size of the partition. In the analysis the size of the constructed partition depends on the probability of the event happening that a line segment $u$ comes before a line segment $v$ in the random permutation $u, \ldots, u_i, v$. Using the the min-wise probabilities of \Cref{tab:minwise} directly we get the same bounds on the partition size as running times on quicksort using the min-wise analysis. This analysis is tightened through \Cref{thm:mainQS} for both $k=2$ and $k=4$. 

By an analogous argument, the randomized treap data structure of \cite[Page 201]{Motwani:1995:RA:211390} gets using the min-wise bounds expected node depth $\Oh ( \log n)$ when a treap is built over a size $n$ set. Under limited independence using the min-wise analysis, the bounds achieved are then $\{\Oh(\sqrt{n}), \Oh(\log^2 n), \Oh( \log^2 n), \Oh( \log n) \}$ for $k = \{2,3,4,5\}$ respectively. By \Cref{thm:mainQS} we get $\Oh(\log^2 n)$ for $k=2$ and $\Oh( \log n)$ for $k=4$. 
\section{Largest bucket size}
We explore the standard case of throwing $n$ balls into $n$ buckets using a random hash function. We are interested in analyzing the bucket that has the largest number of balls mapped to it. Particularly, for this problem our main contribution is an explicit family of hash functions that are $k$-independent (remember \Cref{def:hashind}) and where the largest bucket size is $\Omega \left( n^{1/k} \right)$. However we start by stating the matching upper bound.
\subsection{Upper bound}\label{sec:load:upper}
We will briefly show the upper bound that matches our lower bound presented in the next section. We are unaware of literature that includes the upper bound, but note that it follows from a standard argument and is included for the sake of completeness.

\begin{lemma}\label{lem:loadupper}
	Consider the setting where $n$ balls are distributed among $n$ buckets using a
	random hash function $h$.
	For $m = \Omega \left ( \frac{\log n}{\log \log n} \right )$ and any $k \in \mathbb{N}$ such that $k < n^{1/k}$
	then if $h$ is $k$-independent the largest bucket size is $\Oh(m)$
	with probability at least $1 - \frac{n}{m^{k}}$.
\end{lemma}
\begin{proof}
Consider random variables $B_1, \ldots, B_n$, where $B_i$ denotes the number of balls that are distributed
to bin $i$. By definition, the largest bucket size is $\max_i B_i$
Since $(\max_i B_i)^{\underline{k}} \leq \sum_{i} (B_i)^{\underline{k}}$ for any threshold $t$ we see that
\begin{align*}
	\Pr( \max_i B_i \geq t ) =
	\Pr\left( (\max_i B_i)^{\underline{k}} \geq t^{\underline{k}} \right) \leq
	\Pr \left( \sum_i (B_i)^{\underline{k}} \geq t^{\underline{k}} \right) \text{.}
\end{align*}
Since $\sum_i (B_i)^{\underline{k}}$ is exactly the number of ordered $k$-tuples being assigned to the same
bucket we see that $\E \left ( \sum_i (B_i)^{\underline{k}} \right ) = n^{\underline{k}} \cdot \frac{1}{n^{k-1}}$,
because there are 
exactly $n^{\underline{k}}$ ordered $k$-tuples. Hence we can apply Markov's inequality
\[
	\Pr \left( \sum_i (B_i)^{\underline{k}} \geq t^{\underline{k}} \right) \leq
	\frac{\E \left( \sum_i (B_i)^{\underline{k}} \right)}{t^{\underline{k}}} =
	\frac{n^{\underline{k}}}{n^k} \cdot \frac{n}{t^{\underline{k}}}
	\le
	\frac{n}{t^{\underline{k}}}
	\text{.}
\]
Since $k < n^{1/k}$ implies $k = \Oh \left ( \frac{\log n}{\log \log n} \right )$ we see that $k + m = \Theta(m)$.
Letting $t = k + m$ we get the desired upper bound $\frac{n}{m^k}$ on the probability that
$\max_i B_i \geq m + k$ since $(m+k)^{\underline{k}} > m^k$.
\qed
\end{proof}

\subsection{Lower bound}\label{sec:load:lower}

At a high level, our hashing scheme is to divide the buckets into sets of size $p$ and in each set polynomial hashing is used on the keys that do not ``fill'' the set. The crucial point is then to see that for polynomial hashing, the probability that a particular polynomial hashes a set of keys to the same value can be bounded by the probability of all coefficients of the polynomial being zero. Having a bound on this probability, the set size can be picked such that with constant probability the coefficients of one of the polynomials is zero, resulting in a large bucket.

\begin{proof}(of \Cref{thm:loadmain})
	Fix $n$, $m$, and $k$.
	We will give a scheme to randomly choose a vector $x = (x_0,\ldots,x_{n-1}) \in [n]^n$
	such that the entries are $k$-independent.
	
	First we choose some prime
	$p \in \left [ \frac{1}{4} m,
	               \frac{1}{2} m \right ]$.
	This is possible by Bertrand's postulate.

	Let $t = \floor{\frac{n}{p}}$ and partition $[n]$ into $t+1$ disjoint sets
	$S_0,S_1,\ldots,S_{t}$, such that $\abs{S_i} = p$ when $i < t$ and $\abs{S_t} = n-pt = (n \bmod p)$.
	Note that $S_t$ is empty if $p$ divides $n$.
	
	The scheme is the following:
	\begin{itemize}
		\item First we pick $t$ polynomial hash function $h_0,h_1,\ldots,h_{t-1} : [p] \to [p]$
		      of degree $k$, i.e.
		      $h_i(x) = a_{i,k-1}x^{k-1} + \ldots + a_{i,0} \bmod p$ where $a_{i,j} \in [p]$
		      is chosen uniformly at random from $[p]$.
		\item For each $x_i$ we choose which of the events $(x_i \in S_0), \ldots, (x_i \in S_t)$
		      are true such that $P(x_i \in S_j) = \frac{\abs{S_j}}{n}$. This is done independently
		      for each $x_i$.
		\item For each $j = 0,\ldots,t-1$ we let $Y_j = \set{x_i \mid x_i \in S_j}$ be the
		      set of all $x_i$ contained in $S_j$. If $\abs{Y_j} > p$ we let $Z_j \subset Y_j$
		      be a subset with $p$ elements and $Z_j = Y_j$ otherwise.
		      We write $Z_j = \set{x'_0,\ldots,x'_{r-1}}$ and $S_j = \set{s_0,\ldots,s_{p-1}}$.
		      Then we let $x'_{\ell} = s_{h_i(\ell)}, \ell \in [r]$. The values for $Y_j \backslash Z_j$
		      are chosen uniformly in $S_j$ and independently.
		\item For all $x_i$ such that $(x_i \in S_t)$ we uniformly at random and independently
		      choose $s \in S_t$ such that $x_i = s$.
	\end{itemize}
	This scheme is clearly $k$-independent. The at most $p$ elements in $Y_j$ we distribute using a $k-1$
	degree polynomial are distributed $k$-independently as degree $k-1$ polynomials over $p$ are known to be $k$-independent (see e.g.~\cite{joffe1974}).
	The remaining elements are distributed fully independently.
	
	We can write $\abs{S_i} = \sum_{j=0}^{n-1} [x_j \in S_i]$ and therefore $\abs{S_i}$ is the sum of
	independent variables from $\set{0,1}$. Since
	$\E \left ( \abs{S_i} \right ) = p = \omega(1)$ a standard Chernoff bound gives us that
	\begin{align}
		\label{eq:atLeastHalf}
		\Pr \left (
			\abs{S_i}
			\le 
			\left ( 1- \frac{1}{2} \right )p
		\right )
		\le 
		e^{-\Omega(p)} = o(1)
		\ \text{.}
	\end{align}
	For $i \in [t]$ let $X_i$ be $1$ if $S_i$ consists of at least $p/2$ elements and $0$ otherwise.
	In other words $X_i = \left [ \abs{S_i} \ge p/2 \right ]$. By \Cref{eq:atLeastHalf} we see that
	$\E (X_i) = 1-o(1)$. Let $X = \sum_{i=0}^{t-1} X_i$. Then $\E(X) = t(1-o(1))$, so we can 
	apply Markov's inequality to obtain
	\begin{align*}
		\Pr \left (
			X \le \frac{1}{2}t
		\right ) = 
		\Pr \left (
			t - X \ge \frac{1}{2}t
		\right ) \le 
		\frac{\E (t - X)}{\frac{1}{2}t} =
		o(1)
		\ \text{.}
	\end{align*}
	So with probability $1 - o(1)$ at least half of the sets $S_i, i \in [t]$ contain at least $p/2$
	elements. Assume that this happens after we for every $x_i$ fix the choice of $S_j$ such that
	$x_i \in S_j$, i.e. assume $X \ge t/2$.
	Wlog. assume that $S_0, \ldots, S_{\ceil{t/2}-1}$ contain at least $p/2$ elements.
	For each $j \in \left [ \ceil{t/2} \right ]$ let $Y_j$ be $1$ if $h_j$ is constant and $0$ otherwise.
	That is, $Y_j = \left [ a_{i,k-1} = \ldots = a_{i,1} = 0 \right ]$. We note that $Y_j$ is $1$ with
	probability $\frac{1}{p^{k-1}}$. Since $Y_0,\ldots,Y_{\ceil{t/2}-1}$ are independent we see that
	\begin{align*}
		\Pr \left (
			Y_0 + \ldots + Y_{\ceil{t/2}-1} > 0
		\right )
		& = 
		1 - \left ( 1 - \frac{1}{p^{k-1}} \right )^{\ceil{t/2}}		
		\\ &		
		\ge 
		1 - e^{-\frac{\ceil{t/2}}{p^{k-1}}}
		=
		1 - e^{-\Theta(n/p^k)}
	\end{align*}
	Since $p \le m$ we see that $e^{-\Theta(n/p^k)} \le e^{-\Theta(n/m^k)}$ furthermore $n/m^k \le 1$ by
	assumption and so $e^{-\Theta(n/m^k)} = 1 - \Theta \left ( \frac{n}{m^k} \right )$. This proves that
	at least one $h_i, j \in \left [ \ceil{t/2} \right ]$ is constant with probability
	$\Omega \left ( \frac{n}{m^k} \right )$. And if that is the case at least on bucket has size
	$\ge p/2 = \Omega(m)$.
	This proves the theorem under the assumption that $X \ge t/2$. Since $X \ge t/2$ happens with probability
	$1 - o(1)$ this finishes the proof.

\qed
\end{proof} 

Since it is well known that using $\mathcal{O}(\log n / \loglog n)$-independent hash function to distribute the balls will imply largest bucket size $\Omega \left( \log n / \log \log n \right )$ , \Cref{cor:full} provides the full understanding of the largest bucket size.
\begin{proof}(of \Cref{cor:full})
	Part (a) follows directly from \Cref{thm:loadmain}. Part (b) follows since $k > n^{1/k}$ implies $k > \log n / \log \log n$ and so we apply the $\Omega \left( \log n / \log \log n \right )$ bound from \cite{Schmidt:1995:CBA:217737.217746}.
\qed
\end{proof} 


\bibliographystyle{amsplain}
\bibliography{kind}

\appendix
\section{Appendix}

\subsection{Quicksort in Setting 2}\label{sec:settingtwo}

The analog to \Cref{lem:boundC4} that we need in order to prove that quicksort
in Setting 2 using a $4$-independent hash function runs in expected
$\Oh \left ( n \log n \right )$ time is proved below.

\begin{lemma}
Let $h: X \to (0,1)$ be a $4$-independent hash function and $A,B \subset X$
disjoint sets such that $\abs{A} \le \abs{B}$. Then
\[
	\E \left (
		\abs{
			\set{
				a \in A
				\mid
				h(a) < \min_{b \in B} h(b)
			}
		}
	\right )
	=
	O \left (
		1
	\right )
\]
\end{lemma}
\begin{proof}
Wlog assume that $|A|=|B|= n$. Let $Y$ be defined by
\[Y = \set{a \in A \mid h(a) < \min_{b \in B} h(b) }\text{.}\] If $a \in Y$ then either
$h(a) < \frac{1}{n}$ or there exists $k \in \mathbb{N}$ such that $h(a) \le 2^{-k+1}$
and $\min_{b \in B} h(b) \ge 2^{-k}$, where we can choose $k$ such that $2^{-k+1} > \frac{1}{n}$,
i.e. $2^k < 2n$.

Let $Y_k$ be the set of all keys $a \in A$
satisfying $h(a) \le 2^{-k+1}$ and let $E_k$ be the event that $\min_{b \in B} h(b) \ge 2^{-k}$. Also let $1_{E_k}$ denote the indicator variable defined as being
$1$ when event $E_k$ occurs and $0$ otherwise.
Since the expected number of keys in $A$ hashing below $\frac{1}{n}$ is $1$ we see that:
\[
	\E( \abs{Y})
	\le
	1 + \sum_{k = 1}^{\floor{\lg n} + 1} \E ( \abs{Y_k} 1_{E_k} )
\]
Now note that:
\begin{equation}\label{eq:expectationbound}
	\E( \abs{Y_k} 1_{E_k})
	\le
	\E (\abs{Y_k} - 2^{-k+2}n)^+
	+
	2^{-k+2} n \cdot  \E (1_{E_k})
\end{equation}
where $x^+$ is defined as $\max\set{x,0}$.

First we will bound $\E (\abs{Y_k} - 2^{-k+2}n)^+$. Note that for any $\ell \in \mathbb{N}$:
\begin{align}\label{eq:term1}
	\Pr\left ( (\abs{Y_k} - 2^{-k+2}n)^+ \ge \ell \right )
	&=
	\Pr \left ( \abs{Y_k} - \E(\abs{Y_k}) \ge 2^{-k+1}n + \ell \right )\notag\\
	&\le
	\frac{ \E \left ( \abs{Y_k} - E \abs{Y_k} \right )^4}{(2^{-k+1}n + \ell)^4}
\end{align}
Remember that we consider a $4$-independent hash function $h$. Next we wish to upper bound $\E \left ( \abs{Y_k} - \E (\abs{Y_k}) \right )^4$ (the numerator of (\ref{eq:term1})). Consider indicator variables $X_a$ for all $a \in A$ such that $X_a = 1$ if $a \in Y_k$ and $0$ otherwise. By the definition of $Y_k$ we have $|Y_k| = \sum_{a \in A} X_a$ and $\E(\sum_{a \in A} X_a) = \mathcal{O}(2^{-k+1})$.
\begin{align}
	\E \left ( \abs{Y_k} - \E (\abs{Y_k}) \right )^4
	&=
    \E \left( \sum_{a \in A} X_a - \E(X_a) \right)^4\notag\\
    &=
    \mathcal{O}\left(n \E(X_a - \E(X_a))^4 + n^2  \E( (X_a - \E(X_a))^2)^2 \right)\notag\\
    &=
    \mathcal{O}\bigl( \E \bigl(\sum_{a \in A} X_a\bigr)^2 \bigr)
    =  \mathcal{O}((2^{-k} n)^2) \label{eq:numbound}
\end{align}
Consider now the following fact, which we will use to bound a particular type of sum.
\begin{fact}\label[fact]{fact:funkysum}
Let $r,l \in \R$. It holds that \[ \sum_{l \geq 1} \frac{1}{(r+l)^4} \leq \frac{1}{r^3} \text{.}\]
\end{fact}
\begin{proof}
We have \[ \sum_{l \geq 1} \frac{1}{(r+l)^4} \leq  \int_0^{\infty} \frac{1}{(r+x)^4}\mathrm{d}x = \left[ - \frac{1}{3} \frac{1}{(r+x)^3} \right]_0^\infty \leq  \frac{1}{r^3} \text{.}\]
\qed
\end{proof}
By application of \Cref{fact:funkysum} and using our bound from (\ref{eq:numbound}) we can finish the upper bound on (\ref{eq:term1}):
\begin{align*}
	\E (\abs{Y_k} - 2^{-k+2}n)^+
	& =
	\sum_{\ell \ge 1}
	\Pr\left ( (\abs{Y_k} - 2^{-k+2}n)^+ \ge \ell \right )
	\\
	& =
	\mathcal{O} \left (
		(2^{-k} n)^2
		\sum_{\ell \ge 1} \frac{1}{(2^{-k+1}n + \ell)^4}
	\right ) \\
	& =
	\mathcal{O} \left (
		\frac{1}{2^{-k}n}
	\right )
	=
	\mathcal{O} \left (
		\frac{2^k}{n}
	\right )
\end{align*}

We only need to bound $2^{-k+2} n \cdot  \E( 1_{E_k})$ (the second term of (\ref{eq:expectationbound})) in order to
finish the proof.
For each $b \in B$ let $Z_b = 1$ if $h(b) \le 2^{-k}$ and $Z_b = 0$ otherwise. Then
$E_k$ implies that $\sum_{b \in B} Z_b = 0$. Let $Z = \sum_{b \in B} Z_b$. Then by an equivalent argument as used  for (\ref{eq:numbound}):
\[
	\E (1_{E_k})
	=
	\Pr \left (
		Z = 0
	\right )
	\le
	\Pr \left (
		\abs{Z - \E Z} \ge \E Z
	\right )
	\le
	\frac{E (Z - \E Z)^4}{(\E Z)^4}
	=
	\Oh \left ( \frac{1}{(\E Z)^2} \right )
\]
Since $\E(Z) = 2^{-k} n$ we see that
\[
	2^{-k+2} n \cdot  \E( 1_{E_k})
	=
	\mathcal{O} \left (
		\frac{1}{2^{-k}n}
	\right )
	=
	\mathcal{O} \left (
		\frac{2^k}{n}
	\right )
\]
To conclude, we insert our bounds on the two terms of (\ref{eq:expectationbound}), which completes the proof.
\begin{align*}
	\E (\abs{Y})
	&\le
	1
	+
	\sum_{k = 1}^{\floor{\lg n} + 1}
	\E (\abs{Y_k} - 2^{-k+2}n)^+
	+
	2^{-k+2} n \cdot \E (1_{E_k})\\
	&=
	1+
	\sum_{k = 1}^{\floor{\lg n} + 1}
	\mathcal{O} \left (
		\frac{2^k}{n}
	\right )
	=
	\mathcal{O}(1)
\end{align*}
\qed
\end{proof} 


\end{document}